\documentclass[journal]{IEEEtran}

\usepackage[latin1]{inputenc}
\usepackage{amsmath}
\usepackage{amsthm}
\usepackage{amsfonts}
\usepackage{amssymb}
\usepackage{hyperref}
\usepackage{cleveref}
\usepackage{graphicx}
\usepackage{nicefrac}
\usepackage{fullpage}
\usepackage{bm}
\usepackage{xcolor}
\usepackage[linesnumbered,ruled]{algorithm2e}
\usepackage{tikz}
\usepackage{mathtools}
\usepackage[all]{xy}
\usepackage{authblk}

\newtheorem{theorem}{Theorem}[section]
\newtheorem{conjecture}[theorem]{Conjecture}
\newtheorem{proposition}[theorem]{Proposition}

\newtheorem{corollary}[theorem]{Corollary}

\newtheorem{definition}[theorem]{Definition}

\newtheorem*{remarkx}{Remark}

\DeclareMathOperator{\lcm}{lcm}

\newcommand{\bbN}{\mathbb{N}}
\newcommand{\bbNp}{\mathbb{N}^*}

\newcommand{\ef}{\mathbb{F}}

\newcommand{\efq}{\ef_q}
\newcommand{\efqm}{\ef_{q^m}}
\newcommand{\eftwo}{\ef_2}
\newcommand{\eftwomm}[1]{\ef_{2^{#1}}}
\newcommand{\eftwom}{\eftwomm{m}}

\newcommand{\bs}{\boldsymbol}

\newcommand{\bch}{\mathrm{BCH}}

\newcommand{\deltading}{{2^{m-2}+1}}
\newcommand{\dbding}{{\frac{2^m+1}{3}}}

\newcommand{\Fcyc}{F^{(\mathrm{cyc})}}
\newcommand{\fcyc}{f^{(\mathrm{cyc})}}

\title{The Generating Idempotent Is a Minimum-Weight Codeword for Some
Binary BCH Codes}

\author{ 
\IEEEauthorblockN{Yaron Shany and Amit Berman}
\thanks{Yaron Shany and Amit Berman are with Samsung  Semiconductor
Israel R\&D Center, 146 Derech Menachem Begin St., Tel Aviv 6492103,
Israel. Emails: \{yaron.shany, amit.berman\}@samsung.com}
}

\begin{document}

\maketitle

\begin{abstract}
In a paper from 2015, Ding {\it et al.} (IEEE Trans.~IT, May 2015)
conjectured that for odd $m$, the minimum distance of the binary BCH
code of length $2^m-1$ and designed distance $2^{m-2}+1$ is equal to
the Bose distance calculated in the same paper. 

In this paper, we prove the conjecture. In fact, we prove a stronger
result suggested by Ding {\it et al.}: the weight of the
generating idempotent is equal to the Bose 
distance for both odd and even $m$. Our main tools are some new
properties of the so-called fibbinary integers, in particular, the
splitting field of related polynomials, and the relation of these
polynomials to the idempotent of the BCH code.
\end{abstract}

\section{Introduction}
BCH codes are widely used in storage systems and various communication
systems, including optical communications, digital video broadcasting,
and more. Despite being introduced more than sixty years ago \cite{Hocq},
\cite{Bose}, the exact minimum distance of BCH codes is known only for a
small number of cases; we refer the reader to the recent
survey \cite{DingLi24} for an account of the current known results. 
The main objective of the current paper is to resolve the minimum
distance for a case that was left as a conjecture in \cite{DDZ15}.

More specifically, for positive integers $m,d$ with $d$ odd, let
$\bch(m,d)$ be the 
(primitive, narrow-sense) binary BCH code of length $2^m-1$ and
designed distance $d$, that is, the binary cyclic code of 
length $2^m-1$ having zeros exactly
$\alpha,\alpha^2,\ldots,\alpha^{d-1}$ and their conjugates, where 
$\alpha$ is a primitive element of $\eftwom$.\footnote{For a prime
power $q$, we let $\efq$ be the finite field of $q$ elements.}
Equivalently, $\bch(m,d)$ is the binary cyclic code with generator
polynomial
$\lcm\big\{M_{\alpha^i}(X)|i\in\{1,\ldots,d-1\}\big\}$, 
where for $\beta\in \eftwom$, $M_{\beta}(X)$ is the minimal polynomial
of $\beta$ over $\eftwo$. It is well known by the {\it BCH bound} that
the minimum distance\footnote{Throughout, ``distance'' and ``weight'' 
stand for Hamming distance and Hamming weight, respectively.} of
$\bch(m,d)$ is at least $d$.

It may happen that $\bch(m,d_1)=\bch(m,d_2)$ for $d_2> d_1$:
starting with $d_1$ and including all conjugates of
$\alpha,\alpha^2\ldots,\alpha^{d_1-1}$ as zeros, it may turn out that
the longer progression $\alpha,\alpha^2,\ldots,\alpha^{d_2-1}$ appears
in the set of zeros. Starting with a {\it designed distance} of
$\delta$, the largest odd $d$ for which $\bch(m,\delta)=\bch(m,d)$ is
called the {\bf Bose distance}. So, the minimum distance of
$\bch(m,\delta)$ is \emph{at least} the Bose distance, but it is
possible that the minimum distance of $\bch(m,\delta)$ is larger than
the Bose distance. In general, determining the minimum distance of BCH
codes is a notoriously hard problem \cite{DingLi24}. 

In a paper from 2015 \cite{DDZ15}, Ding {\it et al.} found the Bose
distance $d_B$ of $\bch(m,\delta)$ for $\delta:=\deltading$, proved that for
even $m$ it is equal to the minimum distance, and posed the following
conjecture:

\begin{conjecture}[Ding {\it et al.}, Conjecture 1 of
\cite{DDZ15}]\label{conj:ding}  
Let $m$ be odd. Then the minimum distance $d$ of $\bch(m,\deltading)$
equals the Bose distance, that is,
$$
d=d_B=\dbding.
$$
\end{conjecture}

It is also mentioned in \cite{DDZ15} that to prove the conjecture, it is
sufficient to find some codeword whose weight equals $d_B$, and that
checking some small values of $m$, it seems that the generating
idempotent of the code has this property. However, the authors of
\cite{DDZ15} state that they were not able to prove this in general,
and invite the reader to attack this open problem
\cite[p.~2355]{DDZ15}.\footnote{It is also stated in \cite{DDZ15} that
\cite{AS94} may be useful for this purpose, but eventually we have
used a different approach.} 

The main result of the current paper is the following theorem. 

\begin{theorem}\label{thm:main}
For all integer $m\geq 4$ (both odd and even), the weight $w$ of the
generating idempotent of $\bch(m,\deltading)$ is equal to the Bose
distance found in \cite{DDZ15}, that is, 
$$
w=d_B=\begin{cases}
\frac{1}{3}(2^m+1) & \text{if $m$ is odd}\\
\frac{1}{3}(2^m-1) & \text{if $m$ is even}.
\end{cases}
$$
\end{theorem}

In particular, Theorem \ref{thm:main} proves Conjecture
\ref{conj:ding}. Moreover, although for the case of even $m$ the
minimum distance was settled in \cite{DDZ15}, it was not previously
known that for even $m$, the generating idempotent attains the minimum
weight.

\subsection{Proof technique}
Our main tools are the relation of the relevant idempotents to
``fibbinary polynomials'' (see Definition \ref{def:fibpol} ahead), as
well as some new properties of these polynomials. In particular, we
specify the splitting fields of the fibbinary polynomials (Proposition
\ref{prop:split}), which seems interesting for its own sake. 

The proof proceeds along the following steps:

\begin{itemize}

\item Identify a polynomial $g(X)\in\eftwo[X]$ 
whose number of roots in  $\eftwom^*$ equals the weight of the
generating idempotent. Hence, the number of roots is at least $d_B$.

\item Show that $g$ factors as $g=Xg_1g_2$, where $\deg(g_1)=d_B$
and $g_2$ is a power of a fibbinary polynomial that splits in
$\eftwomm{m-1}$. 

\item As $\eftwomm{m-1}^*\cap \eftwom^*=\eftwo^*$, $g_2$ may
contribute at most the root $1$ on top of the roots of $g_1$ in
$\eftwom^*$. However, we show that when $g_2(1)=0$, also $g_1(1)=0$,
so that $g_2$ cannot contribute any additional root in
$\eftwom^*$. Hence, the number of roots of $g$ in $\eftwom^*$ is at
most $\deg(g_1)= d_B$. 

\item So, the weight of the generating idempotent must equal $d_B$.

\end{itemize}

\subsection{Paper outline}
Section \ref{sec:prelim} includes some required preliminaries. The
proof of Theorem \ref{thm:main} then appears in Section
\ref{sec:proof} following the above steps: In Section
\ref{sec:nroots}, it is shown that the weight of the generating idempotent
is equal to the number of roots of a ``cyclic fibbinary polynomial''
(Definition \ref{def:fibpol}) in $\eftwom^*$.  Then, in Section
\ref{sec:fact}, it is shown that the cyclic fibbinary polynomial
factors as the product $Xg_1g_2$ described above, and in Section
\ref{sec:split}, the splitting field of $g_2$ is identified. Finally,
all the intermediate results are wrapped up in Section \ref{sec:end}
to complete the proof. We conclude the paper in Section \ref{sec:conc}
with some concluding remarks and open questions for further
research. 

\section{Preliminaries}\label{sec:prelim}
This section includes some notation and definitions that will be used
throughout the paper.

For a positive integer $m$ and for an integer
$i\in\{0,1,\ldots,2^m-1\}$, the $m$-bit {\bf binary representation} of
$i$ is the vector $(i_{m-1},i_{m-2},\ldots,i_0)\in\{0,1\}^m$ such that 
$i=\sum_{j=0}^{m-1} i_j 2^j$. To simplify notation, we write the binary
representation in a string form $i_{m-1}i_{m-2}\cdots i_0$. When
$m$ is fixed, the {\bf most significant bit (MSB)} of $i$ is $i_{m-1}$,
while the {\bf least significant bit (LSB)} of $i$ is $i_0$.

\subsection{Cyclic codes}
Let us first recall some basic definitions related to cyclic
codes. For more details, see, e.g.,
\cite[Chs.~7, 8]{MS} or \cite[Ch.~8]{Roth}. A cyclic code over $\efq$
(for $q$ a prime power) is an $\efq$-linear code that is invariant
under cyclic shifts. Equivalently, a cyclic code of length $n$ is an
ideal in $\efq[X]/(X^n-1)$. By correspondence of ideals of $\efq[X]$
and those of $\efq[X]/(X^n-1)$, any ideal of the latter ring (that is,
any cyclic code) is the image of an ideal of the former ring that
includes $(X^n-1)$, that is, the image of $(g(X))$ for some $g(X)$ with
$g(X)|(X^n-1)$. The {\bf generator polynomial} of a cyclic code
$C\subseteq \efq[X]/(X^n-1)$ is the unique monic $g(X)$ that generates
the pre-image of $C$ in $\efq[X]$ as an ideal.

We will assume that $n$ is coprime to $q$, so that $X^n-1$ is
separable. The {\bf zeros} of a cyclic code $C\subseteq
\efq[X]/(X^n-1)$ are the roots of its generator polynomial,
which all lie in the group of $n$-th roots of unity 
in the splitting field $\efqm$ of $X^n-1$. The set of zeros defines $g$, and
hence determines the code uniquely. We note also that as $g\in
\efq[X]$, its set of roots (i.e., the zeros of the generated cyclic
code) must be invariant under the Galois group of $\efqm/\efq$. Hence
if $\beta$ is a root, so are all its conjugates
$\beta,\beta^q,\beta^{q^2},\ldots$. 

From this point on, we will only consider the binary case, where $q=2$,
and, furthermore, $n:=2^m-1$ for some positive integer $m$. Hence the
zeros of the relevant codes are all in $\eftwom^*$. As already
mentioned in the introduction, we are concerned with primitive,
narrow-sense, binary BCH codes of length $2^m-1$, which are defined as
the binary cyclic codes whose zero set is comprised of
$\alpha,\alpha^2,\ldots,\alpha^{d-1}$ and their conjugates, where
$\alpha\in \eftwom$ is a primitive element, and $d$ is an odd 
integer $\geq 3$. 

\subsection{Cyclotomic cosets}
The orbit of an element $\beta\in \eftwom^*$ under the action of the
Galois group of $\eftwom/\eftwo$ is $\{\beta,\beta^2,\ldots,
\beta^{2^{m'-1}}\}$, where $m'|m$ is the degree of the minimal
polynomial of $\beta$ over $\eftwo$ (equivalently, the extension
degree $[\eftwo(\beta):\eftwo]$). Fixing a primitive element $\alpha$
and writing $\beta=\alpha^i$ for some $i\in\{0,\ldots,n-1\}$, a
convenient way to represent this orbit  
is to record the ``base-$\alpha$ logarithms''
$\{i,2\cdot i\bmod n,2^2\cdot i\bmod n,\ldots,2^{m'-1}\cdot i\bmod
n\}$, where $2^{m'}\cdot i\bmod n=i$, and $m'$ is the smallest positive
integer with this property. Such a set is called the {\bf cyclotomic
coset} of $i$ modulo $n$.

\subsection{Fourier transforms on $\eftwom^*$}
Fix $m\in \bbNp$ and a primitive element $\alpha\in \eftwom$, and
recall that $n=2^m-1$. For a polynomial $f\in \eftwom[X]$ with
$\deg(f)\leq n-1$, we let  the {\bf Fourier transform} of $f$ be the
polynomial\footnote{More precisely, the Fourier transform is the
mapping $f\mapsto\hat{f}$.} 
$$
\hat{f}:=\sum_{i=0}^{n-1} f(\alpha^i)X^i.
$$
It is well-known that $f$ is determined from $\hat{f}$ as
\begin{equation}\label{eq:invF}
f=\sum_{i=0}^{n-1} \hat{f}(\alpha^{-i})X^i
\end{equation}
(this can be shown, e.g., by verifying that the inverse of the
Vandermonde matrix $\{\alpha^{ij}\}_{0\leq i,j\leq n-1}$ is
$\{\alpha^{-ij}\}_{0\leq i,j\leq n-1}$). 
In fact, the Fourier transform, as well as the inversion formula
\eqref{eq:invF}, are well-defined on $\eftwom[X]/(X^n-1)$, since
$\alpha^n=1$. 

\subsection{The generating idempotent of a binary cyclic code}
We continue to restrict attention to the case $n=2^m-1$ of interest to
the paper. It is well known \cite[Ch.~8]{MS} that any 
binary cyclic code $C\subseteq \eftwo[X]/(X^n-1)$ has a unique
codeword\footnote{For simplicity, we identify polynomials of degree
$<n$ with their image in $\eftwo[X]/(X^n-1)$.} $e(X)$ such that
$e^2=e$ in $\eftwo[X]/(X^n-1)$ and also $e$ generates $C$ as an ideal
in $\eftwo[X]/(X^n-1)$. Polynomials $e(X)$ with $e^2=e$ in
$\eftwo[X]/(X^n-1)$ are called {\bf idempotents}, and an idempotent
$e$ that generates a binary cyclic code $C$, as above, is called the
{\bf generating idempotent} of $C$. 

It is easily verified that any idempotent evaluates to either $0$ or
$1$ on $\eftwom^*$, and that the generating idempotent of a cyclic
code $C$ evaluates to $1$ exactly on the set of non-zeros of $C$ in
$\eftwom^*$. Hence, if we fix a primitive element $\alpha\in \eftwom$
and let $N\subseteq \{0,\ldots, n-1\}$ be such that the set of non-zeros of $C$
is $\{\alpha^i|i\in N\}$, then the Fourier transform $\hat{e}$ of the
generating idempotent $e$ of $C$ is given by $\hat{e}=\sum_{i\in N} X^i$.

\section{Proof of Theorem \ref{thm:main}}\label{sec:proof}

Throughout this section, $m$ is a positive integer,
$n:=2^m-1$, and $\alpha$ is a primitive element of $\eftwom$. 

\subsection{Finding a polynomial whose number of roots equals the
weight of the idempotent}\label{sec:nroots}
By definition, the generating idempotent of a cyclic code takes
the value $0$ for any zero of the code, and the value $1$ for any
non-zero of the code. In what follows, for a polynomial $a(X)\in
\eftwom[X]$ of degree up to $n-1$, we write 
$a(X)=\sum_{i=0}^{n-1} a_i X^i$, and let $\bs{a}:=(a_0,\ldots,
a_{n-1})$. 

To continue, it will be useful to define the following sets of
integers. We note that the integers defined in the first part of the
definition are the so-called {\it fibbinary integers} \cite{SFB}.

\begin{definition}
{\rm
\begin{enumerate}
\item Let $F_m$ be the set of integers $i\in \{0,\ldots,n-1\}$ such that the
binary representation of $i$ has no two consecutive
ones. 

\item Let $\Fcyc_m$ be the set of integers $i\in \{1,\ldots,n-1\}$
such that the $m$-bit binary representation of $i$ has no two {\bf
cyclically} consecutive ones.

\end{enumerate}
}
\end{definition}
For example, $F_3=\{0,1,2,4,5\}$, while $\Fcyc_3=\{1,2,4\}$. 

\begin{proposition}\label{prop:max}
It holds that 
\begin{equation}\label{eq:maxfm}
\max{F_m}=\begin{cases}
\frac{2}{3}(2^m+1)-1 & \text{if $m$ is odd}\\
\frac{2}{3}(2^m-1) & \text{if $m$ is even}
\end{cases}
\end{equation}
and
\begin{equation}\label{eq:maxcfm}
\max{\Fcyc_m}=\begin{cases}
\frac{2}{3}(2^m+1)-2 & \text{if $m$ is odd}\\
\frac{2}{3}(2^m-1) & \text{if $m$ is even}.
\end{cases}
\end{equation}
\end{proposition}

\begin{proof}
When $m$ is odd, the $m$-bit binary representation of the largest
integer in $F_m$ is $1010\cdots101$, while that of the largest integer
in $\Fcyc_m$ is $1010\cdots100$. When $m$ is even, the $m$-bit binary
representation of the largest integer in both $F_m$ and $\Fcyc_m$ is
$1010\cdots10$. 
\end{proof}

It will also be useful to define the following polynomials related to
$F_m$ and $\Fcyc_m$.

\begin{definition}\label{def:fibpol}
{\rm
Let $f_m:=\sum_{i\in F_m} X^i\in \eftwo[X]$ and
$\fcyc_m(X):=\sum_{i\in \Fcyc_m} X^i\in \eftwo[X]$. We refer to 
$f_m$ as the $m$-th {\bf fibbinary polynomial}, and to $\fcyc_m$ as
the $m$-th {\bf cyclic fibbinary polynomial}.
}
\end{definition}
Note that $\deg(f_m)=\max{F_m}$ and $\deg(\fcyc_m)=\max{\Fcyc_m}$ are
given by \eqref{eq:maxfm} and \eqref{eq:maxcfm}, respectively.

We will later explore the relations between the fibbinary polynomials
and the cyclic fibbinary polynomials, as well as  useful properties of the
fibbinary polynomials. In particular, we will pin down the splitting field
of the $f_m$. For now, we only relate the cyclic fibbinary polynomials
to the idempotent of $\bch(m,\deltading)$.

From this point on, ``a BCH code'' means a primitive, narrow-sense,
binary BCH code. 

\begin{proposition}\label{prop:wtroots}
Let $e(X)$ be the generating idempotent of $\bch(m,\deltading)$. Then
the weight of $\bs{e}$ equals the number of roots of $\fcyc_m(X)$ in
$\eftwom^*$.  
\end{proposition}

\begin{proof}
It is well-known that for a BCH code of designed distance $\delta$
(for odd $\delta$), all representatives of the conjugacy classes of the zeros
of the code are contained in $\alpha,\alpha^2,\ldots,
\alpha^{\delta-2}$ (as
$\alpha^{\delta-1}=\big(\alpha^{(\delta-1)/2}\big)^2$). So, the
zeros of $\bch(m,\deltading)$ are $\alpha^i$, where $i$ runs on the
cyclotomic cosets of integers $\geq 1$ that are at most
$2^{m-2}-1$. 

Note that the $m$-bit binary representation of
$2^{m-2}-1$ is $0011\cdots1$. This implies
that the set of zeros of $\bch(m,\deltading)$ is exactly
$\{\alpha^i|i\in J_m\}$, where $J_m$ is the set of integers $j$ in
$\{1,\ldots,n-1\}$ such that the binary representation of $j$ has two
cyclically consecutive zeros.\footnote{
If there are two cyclically consecutive zeros, they can be made the
two MSBs by a cyclic shift, resulting in a number $\leq 2^{m-2}-1$. On
the other hand, if there are no two cyclically consecutive zeros, then
for any cyclic shift, at least one of the two MSBs is $1$, so that the
number is $>2^{m-2}-1$.} 
In other words, the set of 
{\bf non}-zeros of the code consists of $\alpha^i$ for $i$ in
$\{0\}\cup \bar{J}_m$, where $\bar{J}_m$ consists of the integers $j$
in $\{1,\ldots,n-1\}$ whose 
binary representation does not have two cyclically consecutive
zeros. Hence  
$$
\hat{e}(X)=1+\sum_{i\in \bar{J}_m} X^i.
$$
The second summand is invariant under squaring modulo $X^n-1$
(as $\bar{J}_m$ is clearly the union of complete cyclotomic cosets), and
therefore evaluates to either $0$ or $1$ on $\eftwom$. 

By \eqref{eq:invF}, the weight of $\bs{e}$ is the number of non-zeros
of $\hat{e}(X)$ in $\eftwom^*$, which, by the above comment, is the
number of zeros of $v(X):=\sum_{i\in \bar{J}_m} X^i$ in
$\eftwom^*$. This number does not change if we replace the exponent
set $\bar{J}_m$ by $\{2^m-1-i|i\in \bar{J}_m\}$, as this change has
the effect of getting $v(\beta^{-1})$ when substituting $\beta\in
\eftwom^*$. But moving from $i$ to $2^m-1-i$ just flips
$0\leftrightarrow 1$ in the $m$-bit binary representation, and
therefore the last set is exactly $\Fcyc_m$, completing the proof.  
\end{proof}

\subsection{Recursions for $f_m(X),\fcyc_m(X)$, and a factorization
of $\fcyc_m(X)$}\label{sec:fact} 
By Proposition \ref{prop:wtroots}, it is of interest to find the
number of roots of $\fcyc_m$ in $\eftwom^*$. Toward this end, we will
study in this section some recurrence relations for the
$f_m(X),\fcyc_m(X)$, as well as their consequence for the required
number of roots.

In what follows, we will use freely the relation
$g(X^{2^i})=g(X)^{2^i}$ for $g(X)\in\eftwo[X]$ and $i\in \bbN$. 

\begin{proposition}\label{prop:cycto}
For $m\geq 4$, it holds that  
\begin{equation}\label{eq:fcyc}
\fcyc_m(X)=Xf_{m-3}(X^4)+1+f_{m-1}(X^2).
\end{equation}
\end{proposition}  

\begin{proof}
Write 
\begin{equation}\label{eq:decomp}
\Fcyc_{m}=A_m\cup B_m,
\end{equation}
where $A_m$, $B_m$ are the subsets consisting of integers whose binary
representation has LSB $0$, $1$, respectively. Then we claim that 
\begin{equation}\label{eq:A}
A_m=\{2j|j\in F_{m-1}\}\smallsetminus\{0\}
\end{equation}
and
\begin{equation}\label{eq:B}
B_m=\{4j+1|j\in F_{m-3}\}.
\end{equation}
To see this, note that the binary representation of the numbers in
$B_m$ must have $2$nd LSB and MSB of $0$ (because the LSB is $1$), and
then the part of the binary representation obtained by omitting the
two LSBs and the MSB can be chosen freely from $F_{m-3}$. A similar
argument justifies also \eqref{eq:A}. Now the assertion follows at
once from \eqref{eq:decomp}--\eqref{eq:B}. 
\end{proof}

Let us now turn to  some recursions for the $f_m$.
\begin{proposition}\label{prop:recur}
It holds that $f_1(X)=1+X$, $f_2(X)=1+X+X^2$. Also, for $m\geq 3$, it
holds that  
\begin{equation}\label{eq:fm1}
f_m=f_{m-1}(X^2)+Xf_{m-2}(X^4),
\end{equation}
alternatively,
\begin{equation}\label{eq:fm2}
f_m=f_{m-1}(X)+X^{2^{m-1}}f_{m-2}(X).
\end{equation}
In addition,
\begin{equation}\label{eq:fm3}
1+f_m=Xf_{m-1}f_{m-2}^2.
\end{equation}

Finally, let $N_m:=|F_m|$ be the the number of monomials appearing in
$f_m$. Then $N_1=2$, $N_2=3$, and for $m\geq 3$,
$N_m=N_{m-1}+N_{m-2}$. Hence  
\begin{equation}\label{eq:fm1z}
f_m(1)=0\quad \iff \quad 3|(m-1).
\end{equation}
\end{proposition}

\begin{proof}
First, \eqref{eq:fm1} follows from decomposing $F_m$ as the union
of the subset of integers having LSB $0$, and the subset having LSB
$1$, and hence two LSBs $01$ (note the similarity to the power
series decomposition in \cite[p.~756]{Arndt}). Also, \eqref{eq:fm2}
follows similarly by considering MSBs instead of LSBs. 

To prove \eqref{eq:fm3}, we use induction. For $m=3$, the equation can
be verified directly. Assume, therefore,
that \eqref{eq:fm3} holds for some $m\geq 3$. Using \eqref{eq:fm1}, we
get
\begin{align*}
1+f_{m+1} &= 1+f_m^2+Xf_{m-1}^4\\ 
 &= X^2f_{m-1}^2f_{m-2}^4+Xf_{m-1}^4 \\
 &\quad \text{ (induction hypothesis)}\\
 &= Xf_{m-1}^2(f_{m-1}^2+Xf_{m-2}^4)\\
 &= Xf_mf_{m-1}^2,
\end{align*}
as required.

Finally, the last assertion follows from either of
\eqref{eq:fm1} or \eqref{eq:fm2}.
\end{proof}

\begin{remarkx}
{\rm
The above Fibonacci recursion for the $N_m$ is a well known property
of the sets $F_m$ of fibbinary integers; it is, in fact, the reason for
the name \cite{SFB}.
}
\end{remarkx}

\begin{corollary}\label{coro:divides}
For $m\geq 4$, it holds that $\fcyc_m=Xu_mf_{m-3}^4$ for some
polynomial $u_m(X)\in \eftwo[X]$ with
\begin{equation}\label{eq:degum}
\deg(u_m)=\begin{cases}
\frac{1}{3}(2^m+1) & \text{if $m$ is odd}\\
\frac{1}{3}(2^m-1) & \text{if $m$ is even}.
\end{cases}
\end{equation}
Hence, for both even and odd $m$, $\deg(u_m)=d_B$, where $d_B$ is the
Bose distance found in \cite{DDZ15}. Moreover, if $f_{m-3}(1)=0$, then
$u_m(1)=0$.
\end{corollary}

\begin{proof}
By \eqref{eq:fcyc}, it holds that
$\fcyc_m=Xf_{m-3}^4+(1+f_{m-1})^2$, while by \eqref{eq:fm3},
$1+f_{m-1}=Xf_{m-2}f_{m-3}^2$. It follows that $\fcyc_m$ is
divisible by $Xf_{m-3}^4$. Now \eqref{eq:degum} follows from
subsisting the degrees of $\fcyc_m$ and $f_{m-3}$ from
\eqref{eq:maxfm} and \eqref{eq:maxcfm}.

For the last assertion, by  substituting \eqref{eq:fm3} in
\eqref{eq:fcyc} as mentioned above, it can be verified that
$u_m=1+Xf_{m-2}^2$, so that $u_m(1)=0$ iff $f_{m-2}(1)=1$. But by
\eqref{eq:fm1z}, this must hold if $f_{m-3}(1)=0$.
\end{proof}

\subsection{The splitting field of the fibbinary polynomials}\label{sec:split}

\begin{proposition}\label{prop:split}
For all integer $m\geq 3$, $f_m$ divides
$X^{2^{m+2}-1}+1$. Specifically, it holds that
\begin{equation}\label{eq:split}
X^{2^{m+2}-1}+1=(1+X)f_m^2+Xf_{m-1}^2f_m+X^3f_{m-1}^4f_m^4.
\end{equation}
Hence for all $m\geq 3$, the splitting field of $f_m$ is
$\eftwomm{m+2}$. 
\end{proposition}

\begin{proof}
We will prove \eqref{eq:split} by induction. For the basis, one can
verify \eqref{eq:split} directly. 

To continue, let us write $t_m$ for the right-hand side of
\eqref{eq:split}. Assume now that $m\geq 4$, and \eqref{eq:split}
holds for $m-1$, that is, $t_{m-1}=X^{2^{m+1}-1}+1$. Since 
$$
X^{2^{m+2}-1}+1=X\big((X^{2^{m+1}-1}+1)+1 \big)^2+1,
$$
it is sufficient to prove that
\begin{equation}\label{eq:mid}
X(t_{m-1}+1)^2+1=t_m.
\end{equation}
Now,
\begin{align*}
X(t_{m-1}+1)^2+1 & =
X(1+X)^2f_{m-1}^4+\underbrace{X^3f_{m-2}^4f_{m-1}^2}_{X(Xf_{m-1}f_{m-2}^2)^2}
\\ &\quad +
\underbrace{X^7f_{m-2}^8f_{m-1}^8}_{X^3f_{m-1}^4(Xf_{m-1}f_{m-2}^2)^4}
+X+1 \\
&= X(1+X)^2f_{m-1}^4+X(1+f_m^2)\\
&\quad +X^3f_{m-1}^4(1+f_m^4)+X+1\\
&\quad \text{ (by \refeq{eq:fm3})} \\
&=
Xf_{m-1}^4+X^3f_{m-1}^4+X+Xf_{m}^2\\
&\quad +X^3f_{m-1}^4+X^3f_{m-1}^4f_m^4+X+1\\
&=1+Xf_{m-1}^4+Xf_m^2+X^3f_{m-1}^4f_m^4.
\end{align*}
The last expression is equal to $t_m$ iff 
$$
1+Xf_{m-1}^4+Xf_m^2=(1+X)f_m^2+Xf_{m-1}^2f_m, 
$$
that is, iff $f_m^2=1+Xf_{m-1}^2(f_m+f_{m-1}^2)$. However, by
\eqref{eq:fm1}, the right-hand side of the last expression equals
$1+X^2f_{m-1}^2f_{m-2}^4$, which, in turn, equals $f_m^2$, by
\eqref{eq:fm3}. This completes the proof of \eqref{eq:split}. 

Hence the splitting field of $f_m$ is a subfield of
$\eftwomm{m+2}$, and $f_m$ is separable. Considering the degree of
$f_m$ from \eqref{eq:maxfm}, the only possible subfield of
$\eftwomm{m+2}$ that can contain all roots is $\eftwomm{m+2}$ itself. 
\end{proof}

\subsection{Completing the proof of Theorem \ref{thm:main}}\label{sec:end}
\begin{proof}[Proof of Theorem \ref{thm:main}]
By Proposition \ref{prop:wtroots}, the weight $w$ of the generating
idempotent is equal to the number of roots of $\fcyc_m$ in
$\eftwom^*$. By Corollary \ref{coro:divides}, $\fcyc_m=Xu_mf_{m-3}^4$,
and by Proposition \ref{prop:split}, all the roots of $f_{m-3}$ are in
$\eftwomm{m-1}$, whose intersection with $\eftwom$ is $\eftwo$. So, the
only root possibly contributed by $f_{m-3}$ in $\eftwom^*$ is $1$, but
by Corollary \ref{coro:divides}, if $1$ is a root of $f_{m-3}$, then
it is already a root of $u_m$.

Hence the total number of roots is at most $\deg(u_m)=d_B$, but also
the number of roots is at least $d_B$, as it is the weight of the
idempotent. 
\end{proof}

\section{Concluding remarks and open questions} \label{sec:conc}
It follows from the proof of Theorem \ref{thm:main} that the
polynomial $u_m(X)$ from the proof of Corollary \ref{coro:divides}
is separable, splits in $\eftwom$, and its roots are $\alpha^i$
exactly for those $i$ for which $X^i$ appears in $e(X)$. Hence
$u_m(X)$ is nothing but the {\it locator polynomial} of the generating
idempotent. Note, however, that in most references, including
\cite{AS94}, the term ``locator polynomial'' refers to the
``reversed'' polynomial $X^{\deg(u_m)}u_m(X^{-1})$, having
$1/\alpha^i$ for $i$ as above as its roots. 

From the proof of Corollary \ref{coro:divides} it follows that
$u_m=1+Xf_{m-2}^2$, so that $u_m$ has only
odd powers of $X$, apart from the free coefficient. Considering that
$u_m$ is reversed with respect to the definition of the locator
polynomial in \cite{AS94}, this agrees with \cite[Theorem 2]{AS94}. 
However, without the above analysis, it is not automatically clear
that $u_m=1+Xf_{m-2}^2$ splits in $\eftwom$, and that it is the
locator of the \emph{generating} idempotent. It would be interesting
to see if the methods of \cite{AS94} can provide an alternative proof
for Theorem \ref{thm:main}.

Finally, we remark that the above analysis also recovers the dimension
of $\bch(m,2^{m-2}+1)$ from \cite[Corollary 10]{DDZ15}. The dimension
is the number of non-zeros of $e(X)$, which equals
$1+|\bar{J}_m|=1+|\Fcyc_m|$, as shown in the proof of Proposition
\ref{prop:wtroots}. From Proposition \ref{prop:cycto}, it can be
verified that $1+|\Fcyc_m|=|F_{m-3}|+|F_{m-1}|$, and as $|F_j|$ is
the $(j+2)$-th Fibonacci number (Proposition \ref{prop:recur}),
Binet's formula for the Fibonacci numbers can be used to recover
\cite[Corollary 10]{DDZ15}; we omit the details.

We conclude the paper with a discussion on
potential generalizations of the current results for determining the
weight of the generating idempotent for other classes of BCH codes.
Consider the case of designed distance
$\delta_{j,m}:=(2^{j-1}-1)2^{m-j}+1=2^{m-1}-2^{m-j}+1$ (where $j\in
\{2,\ldots,m-1\}$), so that $\delta_{2,m}=2^{m-2}+1$ is 
the case considered above, $\delta_{3,m}=3\cdot 2^{m-3}+1$, 
$\delta_{4,m}=7\cdot 2^{m-4}+1$, etc..\footnote{These values of
$\delta_{m,j}$ were chosen because it was observed that at least in
some specific cases, the resulting Bose distance is considerably
larger than the designed distance.}

Let us consider in some 
detail the first new case, that is, the case $j=3$ (with $m\geq 4$),
for which we have $\delta_{3,m}-2=2^{m-1}-2^{m-3}-1$, with $m$-bit
binary representation $01011\cdots 1$. As in the proof of Proposition
\ref{prop:wtroots}, the set of zeros of the code is $\{\alpha^i|i\in
J_{3,m}\}$, where $J_{3,m}$ is the set of integers $i$ in 
$\{1,\ldots,n-1\}$ such that at least one of $000, 001, 010$ appears
on 3 cyclically consecutive indices of the $m$-bit binary representation of
$i$. Hence, the set of non-zeros is $\{0\}\cup \bar{J}_{3,m}$, 
where $\bar{J}_{3,m}$ is the set of integers $i$ in
$\{1,\ldots,n-1\}$ such that neither of $000, 001, 010$ appears on any
3 cyclically consecutive indices of the $m$-bit binary representation of
$i$. This is the same 
as demanding that in any $3$ cyclically consecutive indices, there are
at least two $1$'s.\footnote{Since if $100$ appears on some 3
cyclically consecutive indices, then so does $00b$ for some
$b\in\{0,1\}$. In either case, one of forbidden triples $001$ or $000$
also appears on some 3 cyclically consecutive indices.} 

Moving to inverses, as in the proof of Proposition \ref{prop:wtroots}
(i.e., replacing $i$ with $2^m-1-i$), we conclude that the weight
of the generating idempotent is the number of roots of
$\fcyc_{3,m}(X):=\sum_{i\in \Fcyc_{3,m}} X^i$, where $\Fcyc_{3,m}$ is
the set of integers $i$ in $\{1,\ldots,n-1\}$ such that in every $3$
cyclically consecutive indices of the $m$-bit binary representation,
there is at most a single $1$. As in Section \ref{sec:nroots}, it is
interesting to consider also the corresponding non-cyclic set
$F_{3,m}$, which is the set of integers $i$ in $\{0,\ldots,n-1\}$ with
no more than one $1$-bit in each bit triple, and related polynomials
$f_{3,m}(X):=\sum_{i\in F_{3,m}} X^i$. 

We note that the numbers in
$F_{3,m}$ appear in the OEIS as Sequence A048715 \cite{SFB2}. As in
Proposition \ref{prop:cycto}, it can be verified that for $m\geq 6$,  
\begin{multline*}
\fcyc_{3,m}(X)=1+f_{3,m-2}(X^4) + \\ Xf_{3,m-5}(X^8) + 
X^2f_{3,m-5}(X^{16}). 
\end{multline*}
Also, as in Proposition \ref{prop:recur}, it can be verified that
for $m\geq 6$, it holds that
$f_{3,m}(X)=f_{3,m-1}(X^2)+Xf_{3,m-3}(X^8)$ (and
$f_{3,3}(X),f_{3,4}(X),f_{3,5}(X)$ can be found directly). 
Although examining some small values of $m$ suggests that in general
the generating idempotent is not a low-weight codeword in this case,
it is still interesting to see if its weight can be pinned down by
using the above recursions for identifying  appropriate factors of
$\fcyc_{3,m}$, as done above for the case $j=2$.

Finally, it would be interesting to see if any of the current
results can be generalized to primitive, narrow-sense BCH codes over
$\efq$ for a general prime power $q$, as in \cite[Conjecture
2]{DDZ15}, at least for a designed distance of $q^{m-2}+1$.

\section*{Acknowledgment}
The authors are grateful to Ariel Doubchak for many helpful
comments on an earlier draft.

\bibliographystyle{IEEETrans}
	\bibliography{bibliography}

\begin{IEEEbiographynophoto}{Yaron Shany}
received the Ph.D.~degree in Electrical Engineering from Tel Aviv
University in 2004.  He is currently with the Advanced Flash
Solution Lab of Samsung Semiconductor Israel R\&D Center. His research
interests include coding theory and its applications to storage
systems and devices. 
\end{IEEEbiographynophoto}

\begin{IEEEbiographynophoto}{Amit Berman}
(Senior Member, IEEE) received the Ph.D. degree in
electrical engineering from the Technion---Israel Institute of
Technology, Haifa, Israel, in 2013. He is currently with the Advanced
Flash Solution Laboratory, Samsung Electronics Memory Division, as the
Vice President of Research and Development. He has authored over 30
research articles and holds over 100 U.S./Global issued and pending
patents. He was a recipient of several awards, including the Samsung
Contribution Award, the Hershel Rich Innovation Award, the Mitchell
Grant, the HPI Fellowship, and the International Solid-State Circuits
Conference Recognition. 
\end{IEEEbiographynophoto}

\end{document}